\documentclass[12pt]{article}

\usepackage{amssymb,amsmath,amsfonts,eurosym,geometry,ulem,graphicx,caption,color,setspace,sectsty,comment,footmisc,caption,natbib,pdflscape,subfigure,array,hyperref}
\usepackage[utf8]{inputenc}
\usepackage{url}
\usepackage{multirow}
\usepackage{mathtools}
\newcommand{\mR}{\mathbb{R}}

\newcommand{\mc}{\mathcal}

\newcommand{\Tha}{\Theta}
\newcommand{\tha}{\theta}

\renewcommand{\leq}{\leqslant}
\renewcommand{\geq}{\geqslant}
\mathtoolsset{showonlyrefs}

\usepackage{algpseudocode}
\usepackage{algorithmicx}
\usepackage{algorithm}
\usepackage{caption}
\usepackage{eufrak}

\newtheorem{prop}{Proposition}
\newtheorem{Def}{Definition}
\newtheorem{assume}{Assumption}
\newtheorem{coro}{Corollary}

\newenvironment{proof}[1][Proof]{\noindent\textbf{#1.} }{\ \rule{0.5em}{0.5em}}

\newcolumntype{L}[1]{>{\raggedright\let\newline\\arraybackslash\hspace{0pt}}m{#1}}
\newcolumntype{C}[1]{>{\centering\let\newline\\arraybackslash\hspace{0pt}}m{#1}}
\newcolumntype{R}[1]{>{\raggedleft\let\newline\\arraybackslash\hspace{0pt}}m{#1}}

\begin{document}

\begin{titlepage}
\title{Understanding Police Force Resource Allocation using Adversarial Optimal Transport with Incomplete Information}
\author{Yinan Hu\thanks{Tandon School of Engineering, New York University, Brooklyn, NY, 11201} \and Juntao Chen \thanks{Graduate School of Art and Sciences, Fordham University, Bronx, NY, 10458}\and Quanyan Zhu\thanks{the same as 1st author}}
\date{\today}
\maketitle
\begin{abstract}
\noindent Adversarial optimal transport has been proven useful as a mathematical formulation to model resource allocation problems to maximize the efficiency of transportation with an adversary, who modifies the data. It is often the case, however, that only the adversary knows which nodes are malicious and which are not. In this paper we formulate the problem of seeking adversarial optimal transport into Bayesian games. We construct the concept of Bayesian equilibrium and design a distributed algorithm that achieve those equilibria, making our model applicable to large-scale networks.  
\vspace{0in}
\vspace{0in}

\bigskip
\end{abstract}
\setcounter{page}{0}
\thispagestyle{empty}
\end{titlepage}
\pagebreak \newpage
\doublespacing

\section{Introduction}
The problem of resource matching and allocation has a wide application in many areas such as wireless networks \cite{georgiadis2006resource_allocation_network}, data centers \cite{ferris2013src_data_center}. The ultimate goal of resource allocation is to find a plan to distribute resources to the target so as to maximize the resource efficiency. Optimal transport \cite{kantorovich1942translocation} is one of the centralized schemes to find plans for resource matching and allocations \cite{galichon2018ot_economics,villani2008optimal}. 

In a standard network transport scheme, a dispatcher designs the resource allocation plan that maximizes the accumulated utility function for all dispatchers \cite{juntao2021fair_dynamic_ot} under the assumption that the data received from the dispatchers are truthful. This assumption may not always hold when some dispatchers present selfish \cite{felegyhazi2006nash_eq_ad_hoc_networks,raya2008selfish_behavior_dynamic} or even malicious behaviors when submitting their data \cite{buttyan2007security_coop_network}. Malicious nodes modify the data sent to the dispatcher to disrupt the fairness or the efficiency of the resulting transport plan. Previously there are works\cite{juntao2021adv_ot} that develop a game-theoretic formulation to characterize the interactions between the adversary and the dispatcher in resource allocation processes under the assumption that all the information in the game is common knowledge. That is, the dispatcher knows before the game all the possible modifications that the adversary may adopt upon the data, while the adversary is aware of all possible transport plans and they both know what their opponent know, and so on. However, a general game between an adversary and a defender is one with incomplete information (see chapter 9 of \cite{zamir_game_theory_cambridge}) where the defender is not aware of the adversary's the benignity or malice. In a network transport scenario, where there are multiple target nodes that receive resources, the dispatcher may not know which target nodes are benign and truthfully report their data and which ones are malicious with the intent to modify the data along the links connected to them. A malicious target node may alter the data along the link connected to one particular source node or along the ones connected to several source nodes. 

To this end, we propose a Bayesian game formulation to characterize the asymmetric information and adopt the concept of `types', raised in Harsanyi's model \cite{harsanyi1967games_I,harsanyi1968bayesian_games_II} to capture the adversary's attack mode. Each of the adversary's type represent the links on which the data are compromised. The Bayesian game formulation allows the dispatcher to associate a common distribution over all possible types, so his goal is to maximize the expected utility on the distribution over the type space. The adversary aims at minimizing the aggregated utility of all nodes over his possible modifications given his type. The proposed scenario concerns a variety of attacks in resilient resource allocations such as jamming attacks \cite{garnaev2014jamming_attack} and network topology attacks \cite{shao2020_topology_attack}. 

The planner's utility functions are usually monotonic with respect to the `weights' along the links in the network \cite{srikant2013network_optimization}. The amount of resources transported along links with large weight coefficients tend to dominate in the network and thus causes non-smooth transport plans.  Inspired by \cite{cuturi2013sinkhorn}, we add an extra term, named a regularizer \cite{ferradans2014regularized_ot}, in the utility function with a coefficient that leverages between the solving traditional optimal transport problem and preserving the smoothness of the transport plan. In particular, we introduce the entropic  regularizer because the corresponding regularized optimal transport problems can be solved in distributed algorithms and are thus scalable to large networks \cite{genevay2016stochastic}. 

A centralized planning of the transport within the network requires full knowledge of the data on all links between source nodes and target nodes. As a result, the computational complexity for an optimal transport with a dispatcher grows exponentially with the number of participants. Our goal is to design a decentralized algorithm to obtain a resilient transport plan so that such algorithms can be applied in large-scale networks. We obtain the best responses for source nodes through asynchronous dual pricing (ADP) algorithm \cite{huang2006distributedADP} so that every source node only needs to solve its own Bayesian game against target nodes it connects to.    

The contributions of our paper are threefold: 1) we extend the zero-sum game-theoretic framework raised by \cite{juntao2021adv_ot} into a Bayesian game formulation to capture the adversary's private information. We also extend the previous static game into a multistage dynamic game formulation to capture the adversary's action in altering the belief updates;  2) We introduce an entropic regularizer in the dispatchers' utility functions, smoothing the resulting transport plans; 3) we develop a distributed algorithm to solve our Bayesian game model so that our formulation is applicable to large-scale networks, especially those with ephemeral nodes and thus overcomes the curse of dimensionality.   4) we apply our dynamic adversarial optimal transport scheme to help developing an optimal plan of allocating police resources to combat opportunistic criminals. 

The rest of the paper is organized as follows. In section \ref{sec:formulation} we first consider the problem of regularized optimal transport over networks without adversary, after which we introduce the  adversary that compromises the utility function according to his type, which is unknown to the dispatcher. In section \ref{sec:compute_bayesian_equilibrium} we come up with a distributed algorithm that solves the pairs of the best response of strategies. We demonstrate our simple example in section \ref{sec:numerical_demo}. Finally we conclude our paper in section \ref{sec:conclusion}.

\subsection{Related work}
Learning adaptive learning has been a useful tool to study the crime neighborhood patrolling. 

The topic of police resource allocation for patrolling has been attracting attention for a long time \cite{curtin2010police_patrol}. Many methods have been implemented to characterize and study how to design optimal patrol plan to counteract both opportunistic criminals and strategic adversarys. Different optimal patrol plans have different goals.  Authors in  minimizing the response time or maximizing the patrol coverage   The proposed framework aims a

Authors in \cite{zhang2016allocation_opportunity_criminals} formulate the patrolling model with dynamical Bayesian networks (DBN) and computes the optimal parameters by learning the criminal's behaviors. 
Author in \cite{hespanha2000pursuit_evasion} proposes a pursuit-evasion game framework to capture the conflicting relationship between patrolling police and criminals. However the goal of criminals in such a game model is only to evade capture and does not include any reward function from committing crimes.  
Authors in \cite{jiang2013game_Stackelberg_uncertainty} adopts a Stackelberg security games (SSG) with randomized strategies model to compute an optimal patrolling schedule against passive criminals while incorporating uncertainty in the execution of the patrolling plan.  Instead of combatting opportunistic, passive crimes, the proposed framework designs a patrolling plan against strategic, organized criminals, who knows the police's patrolling allocation plan by methodical investigations and undermines the plan strategically to counteract the police.

The formulation of optimal transport with adversaries has been serving as a useful tool to be implemented to resource allocation problems..

Efficient methods in computing the problems over networks include.....\cite{lei_ying2013communication}  

Asychronous dual pricing...

The rest of the work is organized as follows.

\section{Formulation of source allocation without adversaries}
\label{sec:formulation}
Although different experts may have different orders, but in general there is an agreement that the allocation of police forces serve the following purposes: deterrence of crimes; apprehension of criminals; providing on-demand non-crime police services to the neighborhoods; building trust and confidence among neighborhoods and a sense of security. 

The evaluation of the police work may include the change in crime rate; reduction in police response time; an increase in arrests, etc.

In this section, we propose a formulation optimal transport with regularizers over a bipartite network. 
\subsection{Discrete regularized optimal transport over networks}
In this paper we study the resource allocation process over a bipartite network $G$ that has $N$ source nodes and $M$ target nodes. We denote the sets of source nodes and target nodes as $J=\{j_1,j_2,\dots,j_N\},Q=\{q_1,q_2,\dots,q_M\}$, respectively. Let $\mc{E}$ be the set of possible edges or links, each of which connects a source node $j\in J$ and a target node $q\in Q$. We denote $\{j,q\}$ as a generic link connecting a source node $j$ and a target node $q$. Let $Q_j$ be the set of links that begin with the source node $j$. Correspondingly, let $J_q$ be the set of links that end with the target node $q$.

To characterize the connectivity within the bipartite network $G$, we introduce a 0-1 matrix $B$ with its entries $(b_{jq})_{\substack{j\in J\\ q\in Q}}$ defined as 
\begin{equation}
    b_{jq} = \begin{cases}
    1 & r = \{j,q\},\; \text{if}\; \{j,q\}\in\mc{E},\\
    0 & \text{otherwise}. 
    \end{cases}
    \label{eq:01mtx}
\end{equation}

We characterize a generic dispatcher's action $x\in X = \mR^{|J|\times |Q|}$ as a matrix representation of the rates of transportation of resource from source nodes to target nodes, or the transport plan. For convenience, we denote a generic entry, $x_{jq}$ to represent the transport rate along the edge $\{j,q\}\in \mc{E}$. Notice the transport rate should always be non-negative: $x\geq 0$ (here `$\geq$' or `$\leq$' stands for element-wise comparisons) For a typical source node $j$, we  associate its finite capacity $c_j>0$ to be limitation of the amount of the resources distributed from the source node $j$. 

\subsubsection{The dispatcher's utility function}
We now formulate the planner's utility function in a two-step way: we first formulate a series of utility functions $\{u_{jq}\}_{\substack{j\in J\\ q\in Q}}: \mR_+\rightarrow \mR_+$ for the planner on the link $r$ (which connects the source node $j$ and the target node $q$); then we extend $u_{jq}$ to a utility function $U:\mR^{|J|\times |Q|}\rightarrow \mR_+$ defined on all transport plans over the network. 
We have the following requirements for a utility function in terms of the transporting rate $x_{jq}$.  
\begin{assume}[Monoticity, convexity, and smoothness]
For all $j\in J, q\in Q,\;\{j,q\}\in\mc{E}$, the utility function $u_{jq}:\mR\rightarrow \mR$ is monotone increasing, concave, continuously differentiable in terms of $x_{jq}\in\mR$. 
\end{assume}

We consider $u_{jq}$ as a sum of two terms: $u_{jq} = v_{jq} + \lambda h_{jq}$, where $v_{jq}:\mR\rightarrow \mR$ could be interpreted as a `revenue function' and $h_{jq}$ is named as a regularizer and  $\lambda$ is a trade-off parameter that leverages the portion of the two functions upon the utility function. The first term $v_{jq}$ represents the `revenue' that the planner makes by transporting the resource at the rate $x_{jq}$.  Possible candidates for $v_{jq}$\cite{zhang2019src_matching} include: linear functions,  log-concave functions, etc, depending on the speed of the increase of the revenue with respect to the resources. We choose $v_{jq} = m_{jq}x_{jq},\;\forall \{j,q\}\in \mc{E}$. The coefficient $m_{jq}$ reflects the `perception'` of crime activities on the route $\{j,q\}$ and larger coefficients lead to a hot spot upon the crime heat map among the neighborhoods. The function $h_{jq}:\mR_+\rightarrow \mR$ in \eqref{eq:transport_src_target}  serves the role of a regularizer \cite{cuturi2013sinkhorn} to prevent the planner to allocate all of the resources to a specific target node. We impose the following assumption on the regularizer function:
\begin{assume}[Fairness]
The fairness functions $h_{jq},\;\{j,q\}\in\mc{E}$ are monotonically increasing, convex and differentiable.  
\end{assume}

Typical examples of regularizers include the Kullback–Leibler (KL)-divergence, total variation (TV) functions \cite{ferradans2014regularized_discrete_OT}, quadratic functions \cite{essid2018quad_reg_OT_graph}. In this paper, we adopt the entropic regularizer for every link: for all $x_{jq}\in \mR_+$,
\begin{equation}
    {h_{jq}(x_{jq})} := -\lambda {x_{jq}\log x_{jq}},\;\forall j\in J,\; q\in Q.
    \label{eq:entropy_regularizer}
\end{equation}
Now we extend $u_{jq}$ to a utility function of transport rates over all connections in the network, which is denoted as $U$. We assume $U$ is separable and additive in terms of the transport rate in every node in the network, that is, there is no explicit coupling between the transportation rate on any other two links. We set the utility function in terms of transport rate $x$ of the network as follows.
\begin{equation}
    U(x) = V(x) + \lambda H(x):= \sum_{\{j,q\}\in \mc{E}}{v_{jq}(x_{jq}) + \lambda h_{jq}(x_{jq})},
    \label{eq:network_utility}
\end{equation}
where we assume $V,H$, are also both separable and additive in terms of the links in the network  $H(x) = \sum_{\{j,q\}\in\mc{E}}{h_{jq}(x_{jq})},\;V(x) = \sum_{\{j,q\}\in\mc{E}}{v_{jq}(x_{jq})}$.
In an adversary-free scenario, the goal of the discrete optimal transport over a network is to find an optimal transport plan $x^*$ that maximizes the planner's the utility function, with the constraint that the transportation of resources does not overwhelm any target rate.  We can state the problem of seeking an optimal transport plan as an optimization problem \cite{lei_ying2013communication} as follows:

\begin{equation}
\begin{aligned}
    \underset{x\geq 0 }{\max}&\sum_{\{j,q\}\in\mc{E}}{u_{jq}(x_{jq})},\\
    \text{s.t.}&\;\;Bx \leq c. 
    \label{eq:transport_src_target}
\end{aligned}
\end{equation}

Referring to the form of objective function in \eqref{eq:transport_src_target}, we write the formulation of optimization problem that the dispatcher faces as follows:

\begin{equation}
\begin{aligned}
    \underset{x\geq 0}{\max}&\;\sum_{\{j,q\}\in\mc{E}}{m_{jq} x_{jq} - \lambda x_{jq}\log x_{jq}},\\
    \text{s.t.}&\;\;B\;\text{vec}(x)\leq c,
    \label{eq:regular_ot}
\end{aligned}
\end{equation}
The resulting formulation is sometimes called regularized optimal transport \cite{cuturi2013sinkhorn}.

For the purpose of demonstration, we first depict the plot to illustrate the role of `fairness' term. Let the number of source nodes and target nodes to be $|J|=2,|Q| = 3$. Let the `perception' matrix $m$ to be 
\begin{equation}
    m = \begin{bmatrix}
    1 & 3 & 5 \\
    2 & 5 & 1
    \end{bmatrix}.
\end{equation}
Let the capacity (maximum source available for each node) to be $c = \begin{bmatrix} 4 & 3 & 5 \end{bmatrix}^T$. Then as a bipartite network where every source node is connected to every destination node, the connection matrix $B\in \mR^{|J|\times |Q|}$ can be written as 
\begin{equation}
    B = \begin{bmatrix}
    1 & 1 & 1 & 0 & 0 & 0 \\
    0 & 0 & 0 & 1 & 1 & 1
    \end{bmatrix}.
\end{equation}
To solve the dispatcher's optimization problem, we first list the Lagrangian
\begin{equation}
    \mc{L}_j(x,p) = \sum_{q}{m_{jq} x_{jq} - \lambda x_{jq}\log x_{jq}} -p_j(\sum_{q}{x_{jq}}-c_j),
\end{equation}
where we notice that the variables $p_j$ are `dual pricing' variables.  According to \cite{lei_ying2013communication}, we solve the optimal transport plan through the following iterations
\begin{align}
    x^{(k+1)}_{jq}&= \exp\left(\frac{m_{jq} - p^{(k)}_{j}}{\lambda} -1 \right),\;\forall j\in J, q\in Q,  \\
    p^{(k+1)}_j&= p^{(k)}_j - \gamma (\sum_{q}{x_{jq}} - c_j)^+_{p^{(k)}_j}, j\in J,
\end{align} 
where $\gamma$ refers to the step size and 
\begin{equation}
    (x)^+_y = \begin{cases}
    x & y>0; \\ 
    max(x,0) & y=0
    \end{cases}
\end{equation}


\section{Adversarial regularized optimal transport}
\label{sec:adv_OT_reg}

We now introduce the concept of `adversarial organizer', who manages the allocation of opportunistic criminals among neighborhoods to jeopardize the safety of those neighborhoods and thus the undermine the revenue generated by police allocation.  Different from \cite{zhang2016allocation_opportunity_criminals} and \cite{curtin2010police_patrol} that study individual opportunistic criminals,  we treat crime activity as a group phenomenon since the structure of social networks of individuals play an important role in characterizing his delinquent behaviors.

To characterize the asymmetric information structure and the unknown influence on the data in seeking an optimal transport, we consider a Bayesian game $\mc{G}$ between a dispatcher and an adversary. 

\paragraph{The adversary's information structure} The adversary has private information named his type \cite{harsanyi1967games_I}, which can be constructed in the following way. We denote $\theta_{q}=\{1,2\}$ as adversary's type space on the destination node $j$ and a generic type on the target node is denoted as $\tha_q$. We denote $\tha_q=1$ if the criminals at the target node $j$ commits minor crimes and $\tha_q=2$ if they commit major crimes. Minor crimes incurs smaller loss on dispatcher's revenue Then the overall adversary's type $\theta$ can be obtained by computing Cartesian product among all the `sub-types' as follows:
\begin{align}
    \theta =  \underset{q\in Q}{\bigotimes} \;\theta_{q}  ,\;\Theta = \underset{q\in Q}{\bigotimes} \;\Theta_{q},\; \tha\in\Tha,\;\tha_q\in\Tha_q,\;\forall j\in J. 
    \label{eq:construction_type_space}
\end{align}

\paragraph{Beliefs/types} In the Bayesian game setting, the dispatcher does not know the type the target nodes. He develops a common belief $(\mu(\theta))_{\theta\in \Theta}$ , over the adversary's type space \eqref{eq:construction_type_space} that reflects the conditions of the target nodes. We assume the criminal actions taken at every target node $q$ has a unique type $\tha_q \in \{1,2\}$, where $\tha_q = 1$ means minor crime offenders, while $\tha_q = 2$ means major crime offenders. We therefore have 
\begin{align}
    \mu(\tha) = \prod_{q\in Q}{\mu(\tha_q)}.
\end{align}

\paragraph{Action spaces} Similar to the situation in section \ref{sec:formulation}, a typical dispatcher's action can be represented by a transport plan $x\in\mR^{|J|\times |Q|}$, with each row $x_j,\;j\in J$ being the transport rate vector from a specific source node $j$ and $x_{jq}$ as the transport rate along the link $\{j,q\}$. We shall assume that the transport rate is always non-negative and upper bounded because of the limited capacity of resource that every node owes. Referring to the constraints in \eqref{eq:planner_action}, we construct the strategy space for the dispatcher as $X$:
\begin{align} 
 X_j &:= \{x_j\in \mR^{M}\;\Big|x_j\geq 0,\;\sum_{q\in Q}{x_{jq}}\leq c_j\}, \\
    X &:= \underset{j\in J}{\bigotimes} X_j = \{x\in \mR^{N\times M}\;\Big|\;x\geq 0,\;Bx\leq c\},
    \label{eq:planner_action}
\end{align}
where $B$ is the connection matrix described in \eqref{eq:01mtx} and $c$ is the vector of capacities of the source nodes. 

A generic behavioral strategy of the adversary is denoted as $\xi: \Theta \rightarrow \Xi(\theta),\; \xi(\theta) = (\xi_{q}(\tha_q))_{\substack{j\in J \\ q\in Q}}$, whose entries represent the increase of crime perception regarding the neighborhood $\{j,q\}$. Meanwhile, the magnitude of a criminal's influence upon the perception of the crime activity of a neighborhood should be bounded. 


\begin{equation}
     \xi_{q}(\tha_q) :=\begin{cases}
      \left\{\xi_{q}: \tha_q \rightarrow \mR^{d}_+\Big |\; {\;  \xi_{q}(\tha_q)\leq \underline{n}_{q}},\;q\in Q \right\} & \theta_{q}=1, \\
     \left\{\xi_{q}: \tha_q \rightarrow \mR^{d}_+\Big |\; {\; \xi_{q}(\tha_q)\leq \bar{n}_{q},\;q\in Q} \right\}& \tha_{q} = 2,
     \end{cases}
\end{equation}
where $m_{jq}$ is a parameter in the utility functions representing `perception' of crime rate upon every neighborhood. 
The adversary's action space from the network's viewpoint can be written as a direct product of action spaces of every target node. Thus we have

\begin{equation}
\begin{aligned}
\Xi(\theta)&=\bigotimes_{\{j,q\}\in \mc{E}}{\xi_{q}(\tha_q)} = \bigotimes_{j\in J}{\Xi_{j}(\theta_{q})}. 
     \label{eq:adversary_action}
\end{aligned}   
\end{equation}
 There are `hot spots' of crime regions, where gang bangers gather around and organized crimes happen and there are also `green' areas with very few crime reports every year. Dispatching police resources into `hot spots' incurs higher revenue since `hotpot' places require more police to help suppress illegal activities.

\paragraph{Stage Utility functions} 
We now formulate the dispatcher and the adversary's utitlity/cost functions the game.

The dispatcher leverages two purposes: on one hand, he makes a `revenue' by allocating resources along different links; on the other hand, he also induces `fairness' in the allocation process, making sure that the resources transported along different links do not vary too much. He maximizes his expectation of the sum of revenue and the fairness penalty over his belief upon the types of the target nodes. We denote ${U}^{(p)}: X\times \Xi \rightarrow \mR$ as the utility function for the dispatcher given the strategy profile of $(x,\xi)\in X\times \Xi$ as 

\begin{equation}
     U^{(p)}(x, \xi) 
    = \sum_{\substack{\{j,q\}\in\mc{E}_j \\ \tha_q\in \{1,2\}}}{(m_{jq}+ \tha_q\xi_{q}(\theta_q))\mu(\tha_q)x_{jq}}  - \lambda x_{jq}\log x_{jq}.
    \label{eq:planner_utility}
\end{equation}
The adversary's utility function $U^{(a)}: X\times \Xi \times \Tha \rightarrow  \mR$ not only depends on the actions of the adversary and of the dispatcher, but also depends on the adversary's private type, which is a Cartesian product of the `sub-types'` of every target node. According to \cite{becker1968crime_initial}, from the economist's viewpoint the adversary's utility function $U^{(a)}: \Xi\times \Tha\rightarrow \mR$ consists of also two terms: the `revenue' minus the `expectation of punishment'. The `revenue' stands for the `benefits', either financial or social, that criminals directly obtained from committing crimes. We hereby assumes the revenue is negation of the police's revenue obtained from allocation polices forces, that is, the first term in \eqref{eq:planner_utility}.

The terms of `expected punishment' is a product of the probability of getting caught and the cost of punishment (either fine, probation, imprisonment or even execution). We assumes criminals who commit major crimes will benefit a lot more than the criminals who commit minors, but also felony criminals run a higher risk of getting caught as well as a receiving serious punishment. Inspired by \cite{ehrlich1973illegal_probability_caught}, the criminals' probability of getting caught on every neighborhood relies positively on the amount of police resources and negatively on the intensity of committing the crime. So we can express the criminal's expected penalty from committing criminal action $\xi_{q}$ based on the type $\tha_q\in \Tha_q$ as follows:
\begin{equation}
    \nu_{jq}(\xi_{q}(\tha_q)) = c_{jq} \xi_{q}^{-\beta_2}(\tha_{q})x^{\beta_1}_{jq},
\end{equation}
where $\beta_1,\beta_2\in[0,1]$ are coefficients.
For every type of the adversary $\tha\in\Tha$, we thus formulate the adversarial organizer's utility function $U^{(a)}$ as   
\begin{equation}
     U^{(a)}(x,\xi(\theta))= {\sum_{\{j,q\}\in\mc{E}_j}{c_{jq} \xi_{q}^{-\beta_2}(\tha_{q})x^{\beta_1}_{jq}+ (m_{jq}+\tha_q\xi_{q}(\tha_{q}))x_{jq}}},
    \label{eq:adversary_utility}
\end{equation}
 To summarize and adopt the theories of potential games \cite{shapley1996potential} , the adversary faces the following optimization problem:


\begin{equation}
\begin{aligned}
 \underset{\substack{\xi_1\in \Xi(1) \\ \xi_2\in \Xi(2) }}{\min}&\;{\sum_{\{j,q\}\in\mc{E}_j}{c_{jq}(\xi^{-\beta_2}_{jq}(2)+\xi^{-\beta_2}_{jq}(1)) x^{\beta_1}_{jq}+ (m_{jq}+(\xi_{q}(1)+2\xi_{q}(2))x_{jq}}}, \\
 \text{s.t}. &\;\;{\xi_{q}(1)} \leq \underline{n}_{q}, \;q\in Q, \\
 &\;{\xi_{q}(2)} \leq \bar{n}_{q},\;\;q\in Q
\end{aligned}
\label{prob:adversary_opt}
\end{equation}

We can now define the Bayesian game to characterize the adversarial optimal transport as follows:
\begin{Def}[Bayesian game of adversarial transport]
We define the Bayesian game of adversarial transport $\mc{G}$ as a tuple
\begin{equation}
 \mc{G} = \langle\Xi,X,U^{(a)}, U^{(p)},\Theta, \mu \rangle,
 \label{def:Bayesian_games}
\end{equation}
where $\Xi,X$ are action spaces of the adversary and the dispatcher characterized in \eqref{eq:planner_action} and \eqref{eq:adversary_action}, and $U^{(a)},U^{(p)}$ are utility functions of the dispatcher and the adversary defined in \eqref{eq:planner_utility} and \eqref{eq:adversary_utility} respectively. The adversary has a non-singleton type space denoted as $\Theta$ associated with a common belief $\mu$.
\end{Def}

\section{Analysis of Static Game}
\label{sec:compute_bayesian_equilibrium}
In section \ref{sec:formulation} we come up with formulations of a static Bayesian game and a dynamic Bayesian game to characterize the optimal transport under adversarial setting.
We now solve the Bayesian equilibrium concept raised in section. 
\subsection{Existence of a solution}
\subsubsection{\textbf{Bayesian equilibrium}}

The dispatcher aims at maximizing his expected  regularized utility function, while the adversary attempts to maximize his utility function. 
Thus we come up with the definition of the Bayesian equilibrium for regularized adversarial optimal transport problem as follows. 

\begin{Def}[Bayesian equilibrium for adversarial OT]
\label{def:Bayesian_equilibrium}
A Bayesian equilibrium for the game $\mc{G}$ (these are parameters of the game) of adversarial transport is a strategy profile $(x^*, \{\xi^*(\theta)\}_{\theta\in \Theta}),$ such that neither player has incentive to unilaterally deviate from the strategy to increase their expected payoff. That is,
\begin{align}
 &\forall x\in X,\;   U^{(p)}(x,\xi^*)\geq {U^{(p)}(x^*,\xi^*)} \\
 &\forall \theta\in \Theta,\; \xi(\theta)\in \Xi,\; U^{(a)}(x^*,\xi(\theta))\leq U^{(a)}(x^*,\xi^*(\theta)),
\end{align}
\end{Def}

Given the adversary's behavioral strategy $\xi^*$ (but not his type, which is private information for the adversary), the dispatcher faces the following optimization problem:

\begin{equation}
\begin{aligned}
x^* \in \arg\underset{x\in X}{\max}\;{U}^{(p)}(x,\xi^*),
\end{aligned} \label{eq:br_central_planner}
\end{equation}
while given $x^*$ the adversary faces the maximization problem for all $\theta \in \Theta$ from \eqref{prob:adversary_opt}
\begin{equation}
   \xi^*(\theta)\in \arg\underset{\xi(\theta)\in \Xi}{\min}\;{U}^{(a)}(x^*,\xi(\theta)).
   \label{eq:br_adversary}
\end{equation}
To obtain a Bayesian equilibrium concept of the game is equivalent to solving \eqref{eq:br_central_planner} and \eqref{eq:br_adversary} simultaneously. Since the utility functions for both players are convex and their strategy spaces are compact, we can know that there exists a Bayesian equilibrium.

We have the following existence proposition upon an Bayesian Nash equilibrium of the adversarial optimal transport game $\mc{G}$.
\begin{prop}
Let $\mc{G}$ be the Bayesian games mentioned in definition \ref{def:Bayesian_games}. Let  $(x^*, \{\xi^*(\theta)\}_{\theta \in \Theta})$ be a strategy profile meeting the requirements of Bayesian equilibrium in definition \ref{def:Bayesian_equilibrium}. Then such an equilibrium exists in the game $\mc{G}$.
\end{prop}
\begin{proof}
According to Harsanyi \cite{harsanyi1968bayesian_games_II}, it suffices to prove the existence of a Nash equilibrium, which is obtained before both players learn their types. We observe that the best response function for both players \eqref{eq:br_adversary} and \eqref{eq:br_central_planner} are upper semi-continuous. Meanwhile, referring to \eqref{eq:planner_action} and \eqref{eq:adversary_action}, we know that the strategy spaces of both players are compact and convex. So by Kakutani's fixed-point theorem \cite{kakutani1941fixed_point} there exists a Nash equilibrium for the game where the types of the players are not revealed. As a result, there exists a Bayesian equilibrium for the game $\mc{G}$. 
\end{proof}

\subsection{Dynamic adversarial optimal transport}
The crime activities within one region also vary by time as depicted in Figure  We now introduce the dynamic version of the Bayesian game $\mc{G}$ as in definition \ref{def:Bayesian_games} to further study the temporal relations of police resource allocation under the presence of criminal behaviors.

We can define the history of the dispatcher's plans $x^{(t)}\in X^{\otimes t}$ and the history of the criminal's actions $\xi^t \in \Xi^{\otimes t}$ as follows: 
\begin{Def}[History of the criminal's actions and the dispatcher's plans]
For $t\in T$, we denote the history of the criminal's actions $\xi^{(t)}(\tha)\in \Xi^t$ and the history of the dispatcher's plans $x^{(t)}\in X^{\otimes t}$ as follows:
\begin{equation}
    \begin{aligned}
     x^{(t)} & := (x^1,x^2,\dots,x^t), \\
     \xi^{(t)}(\tha) & := (\xi^1(\tha),\xi^2(\tha),\dots,\xi^t(\tha)),\forall \tha\in\Tha
    \end{aligned}
\end{equation}

\end{Def}

The planner's utilities can be written as 
\begin{equation}
     U^{t(p)}(x^t, \xi^t) 
    = \sum_{\substack{\{j,q\}\in\mc{E}_j \\ \tha_q\in \{1,2\}}}{(m_{jq}+ \tha_q\xi^t_{q}(\theta_q))\mu^t(\tha_q)x^t_{jq}  - \lambda x^t_{jq}\log x^t_{jq}}.
    \label{eq:planner_utility}
\end{equation}

The criminal's utilities 
\begin{equation}
    \begin{aligned}
     \underset{\substack{\xi^t(1)\in \Xi \\ \xi^t(2)\in \Xi}}{\min}\; & {\sum_{\{j,q\}\in\mc{E}_j}{c_{jq}(\phi[\xi^t_{q}(1); \xi^{t-1}_{q}(1)]^{-\beta_2}+\phi[\xi^t_{q}(2); \xi^{t-1}_{q}(2)]^{-\beta_2}) x^{\beta_1}_{jq}}} \\
     &+ (m_{jq}+(\phi[\xi^t_{q}(1); \xi^{t-1}_{q}(1)]+\phi[\xi^t_{q}(2); \xi^{t-1}_{q}(2)])x_{jq}, \\
 \text{s.t}. &\;\;\sum_{j\in J}{\xi^t_{q}(1)} \leq \underline{n}_{q}, \;q\in Q \\
 &\;\sum_{j\in J}{\xi^t_{q}(2)} \leq \bar{n}_{q},\;\;q\in Q,
    \end{aligned}
\end{equation}
where $\phi:\mR\rightarrow \mR$ is some smooth thresholding function with the threshold set to be $\xi^{t-1}$. 

\textbf{setting up the threshold function $\phi$} A rough sketch
of $\phi$ should be as Figure \ref{fig:thresholding}. We aim to characterize a `thresholding behavior' of $\xi^t$ upon $\xi^{t-1}$. One choice is as follows:
\begin{equation}
    \phi(\xi^t;\xi^{t-1}) = \xi^{t-1}\mathbf{1}_{\{\xi^t<\xi^{t-1}+\tau\}} + (\xi^{t}-\tau)\mathbf{1}_{\{\xi^{t}>\xi^{t-1}+\tau\}},
\end{equation}

\begin{figure}
    \centering
    \includegraphics[scale=0.4]{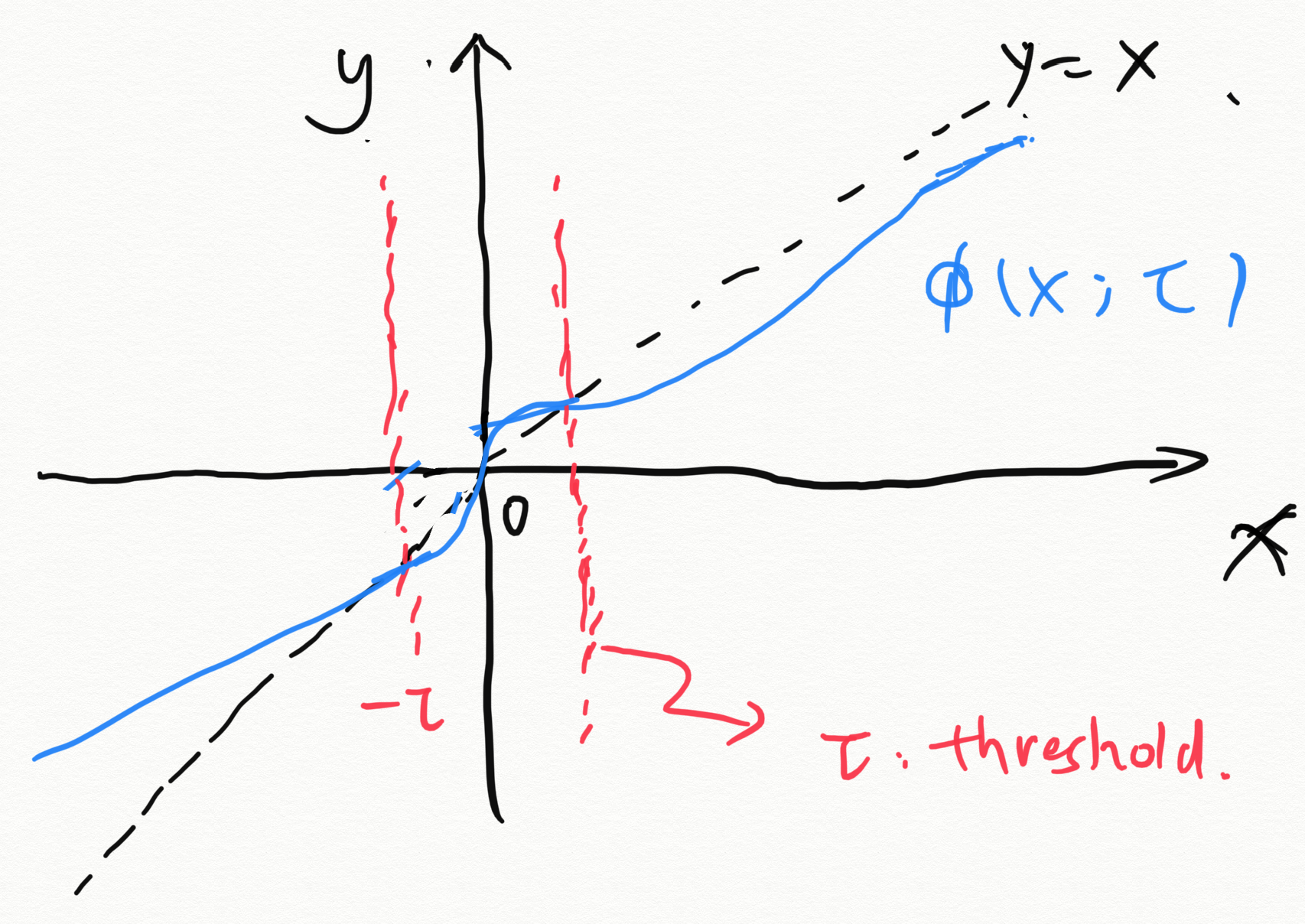}
    \caption{Possible thresholding function $\phi(\xi^t;\xi^{t-1})$. Where $\tau = \xi^{t-1}$ as the criminal's action in the last stage plays the role of threshold. }
    \label{fig:thresholding}
\end{figure}

\subsection{Formulation of dynamic Bayesian game}

\subsection{Solutions}

For a dynamic Bayesian game, many equilibrium concepts can be raised \cite{fudenberg1991game_theory} such as perfect Bayesian Nash equilibrium, sequential equilibrium, and even Bayesian Nash equilibrium. We adopt the following perfect Bayesian Nash equilibrium (PBNE):
\begin{Def}[Perfect Bayesian Nash equilibrium (PBNE)]
Let $\mc{G}^T$ be the dynamic Bayesian game defined in  A perfect Bayesian Nash equilibrium for the dynamic Bayesian game $\mc{G}_T$ consists of the history of strategy profile  $(x^{(t)*},\Xi^{(t)*})$ and a belief system $\mu^*$ such that the following requirements are met:
\begin{enumerate}
    \item (Sequential rationality for the dispatcher) 
    The dispatcher's optimal transport plan $x^{t*}$ at every stage is a maximizer of the dispatcher's utilities given the previous strategy profiles and the current belief $\mu^t$:
    \begin{equation}
    x^t\in\underset{x}{\arg\max}{\;U^{t(p)}(x,\xi^{t*}|x^{(t)*},\xi^{(t)*})}
    \end{equation}
    \item (Sequential rationality for the criminal) The criminal's optimal strategy at every stage $\xi^{t*}(\tha)$ given its type $\tha\in \{1,2\}^{|Q|}$ (the type does not change through stage) is a minimizer of its own cost function $U^{t(a)}$ given the history of the strategy profiles $x^{(t)*},\xi^{(t)*}$ as follows:
    \begin{equation}
        \xi^{t}(\tha)\in \underset{\xi(\tha)\in \Xi}{\arg\max}\; U^{t(a)}(x^{t*},\xi(\tha)|x^{(t)*},\xi^{(t)*}),\;\forall \tha\in\Tha.
    \end{equation}
    \item (Belief update) 
    The belief system of the game is updated after the dispatcher's strategy and the criminal's strategy is reveal at stage $t$:
    \begin{equation}
        \mu^{t+1}(\tha|\xi^t) = \frac{\mu^{t}(\tha)\xi^t(\tha)}{\sum_{\tha'\in\Tha}{{\mu^{t}(\tha')\xi^t(\tha')}}},\;\forall t\in T,\tha \in \Tha.
    \end{equation}
\end{enumerate}
\end{Def}

\section{Distributed Algorithms}

\begin{figure}
   \begin{tabular}{cc}
    \centering
    \includegraphics[scale=0.3]{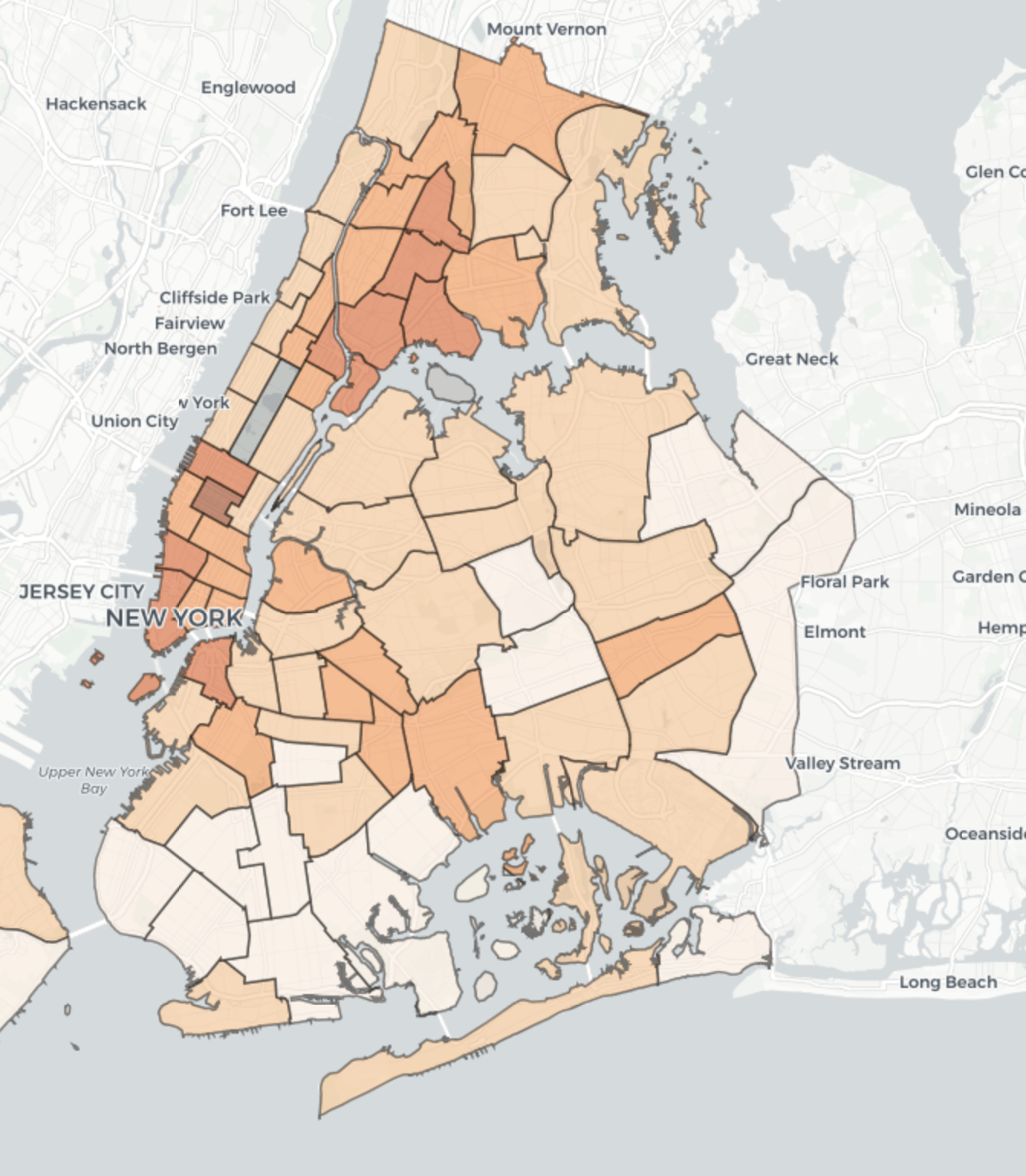}
   &     
    \centering
    \includegraphics[scale=0.32]{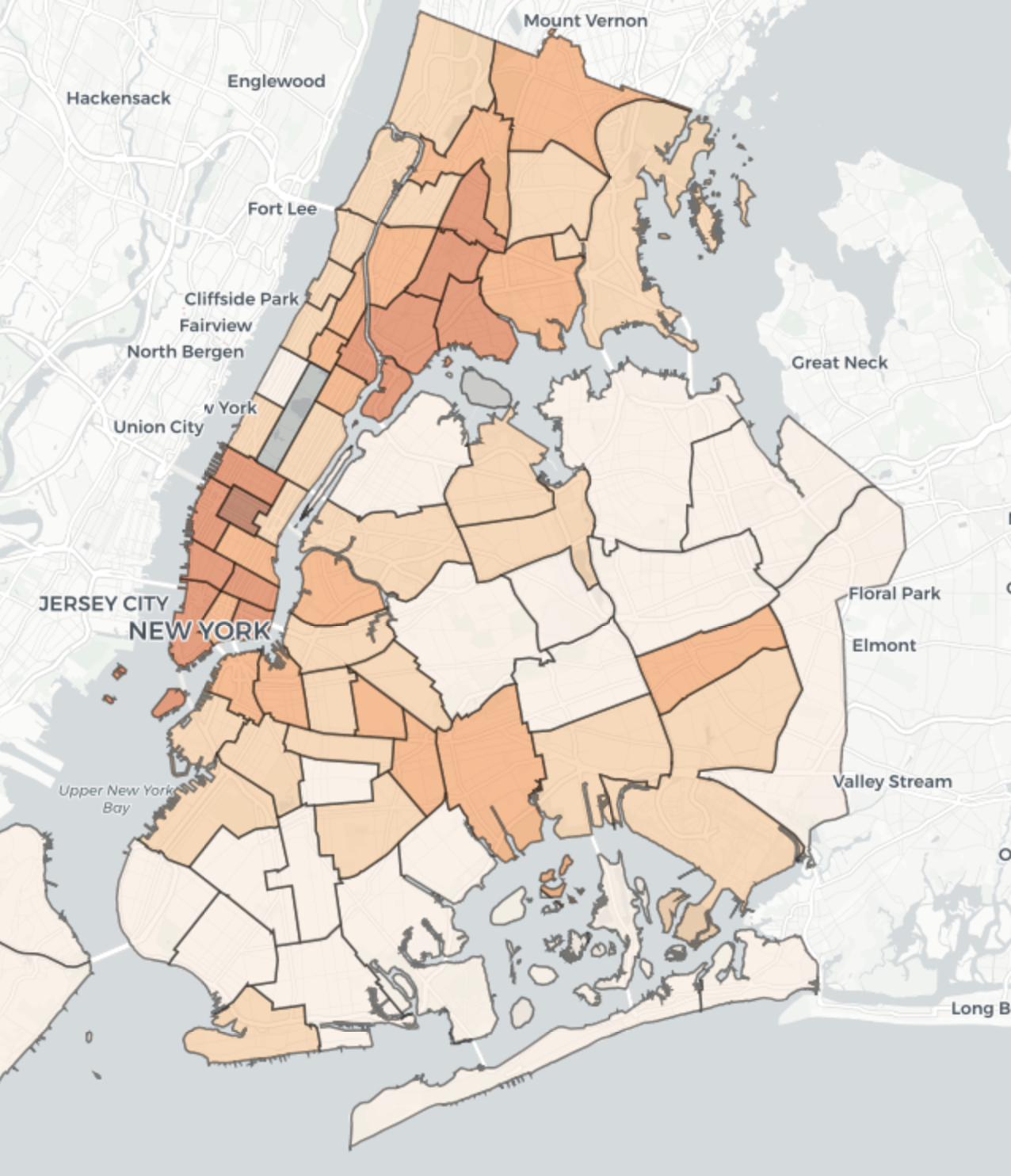}
    \end{tabular}
    \caption{The NYC crime heat map expressed in regions in different time periods. Left: May 2022, Right: June 2022. }
    \label{tab:nyc_crime_map}
\end{figure}

\section{Numerical Experiments}
\label{sec:numerical_demo}
\paragraph*{Crime Neighborhood}
The crime city patrol has been a variety of problems in building smart cities. In particular, a fair, in-time patrol schedule that  allocates police resources among neighborhoods is crucial in detecting and deterring urban crimes.  The most naive way of arranging police forces for patrolling is by planning manual, fixed patrol schedules ahead of time and let the police to commit to the schedule. However fixed patrol schedules do not respond well to emergencies and opportunistic crimes that are committed during the intermittence of shifts. Another branch of models are security pursuit evasion games (PEG)\cite{hespanha2000probabilistic_PEG}, which formulates the police as pursuers and criminals as evaders. However, the main goal of criminals is not just to evade from police, but to benefit from committing crimes.  Another landscape model based on game theories are Stackelberg security games (SSG)\cite{tambe2011Stackelberg_security_game}, in which the criminals are followers and police are leaders. Authors in \cite{pita2011stackelberg_airport} exploits and implements a Stackelberg security game formulation upon the on-demand allocation of security forces among airports. To explore temporal correlations of criminals, authors in \cite{brown2014streets_game_partol} , formulate the patrolling of polices as a Markov decision process (MDP). Authors in \cite{zhang2016patrol_Dynamic_Bayesian_Network_DBN} formulates the patrolling game as a dynamic Bayesian network to better predict the criminals' migration among different neighborhoods. 

All of the models optimizes police force allocation regarding a particular goal such as minimizing response time, lowering the crime rate, etc. In real practice the allocation of police needs to leverage between different goals and so are criminals. Hence we apply our dynamical Bayeisan optimal transport model into formulating police dispatch. We assume  there are two types of criminals: minor criminal offenders and major criminal offenders.Such types are private that the police can only develop a common prior distribution yet do not know the exact realizations. We assume that at every stage, criminals act simultaneously with police patrolling and can adjust their locations of committing crimes at the next moment by observing the police force allocation for the moment. The police forces do not have access to the criminals' actions at previous stages, but criminals do and they can refer to their history of actions to help make decisions for the current stage.  We aim to deploy our formulation of adversarial dynamic optimal transport into crime activities to study temporal and spatial correlation of crimes and how crimes are affect in the presences of police force allocations.

As a concrete example, we present numerical experiments of illustrate the equilibrium achieved both in static and in dynamic Bayesian games to solve optimal transport under the adversarial settings. We consider first a network that contains 2 source nodes and 3 target nodes whose topology is depicted in Figure \ref{fig:transport_plan}.

I 
\begin{figure}
    \centering
    \includegraphics[scale=0.5]{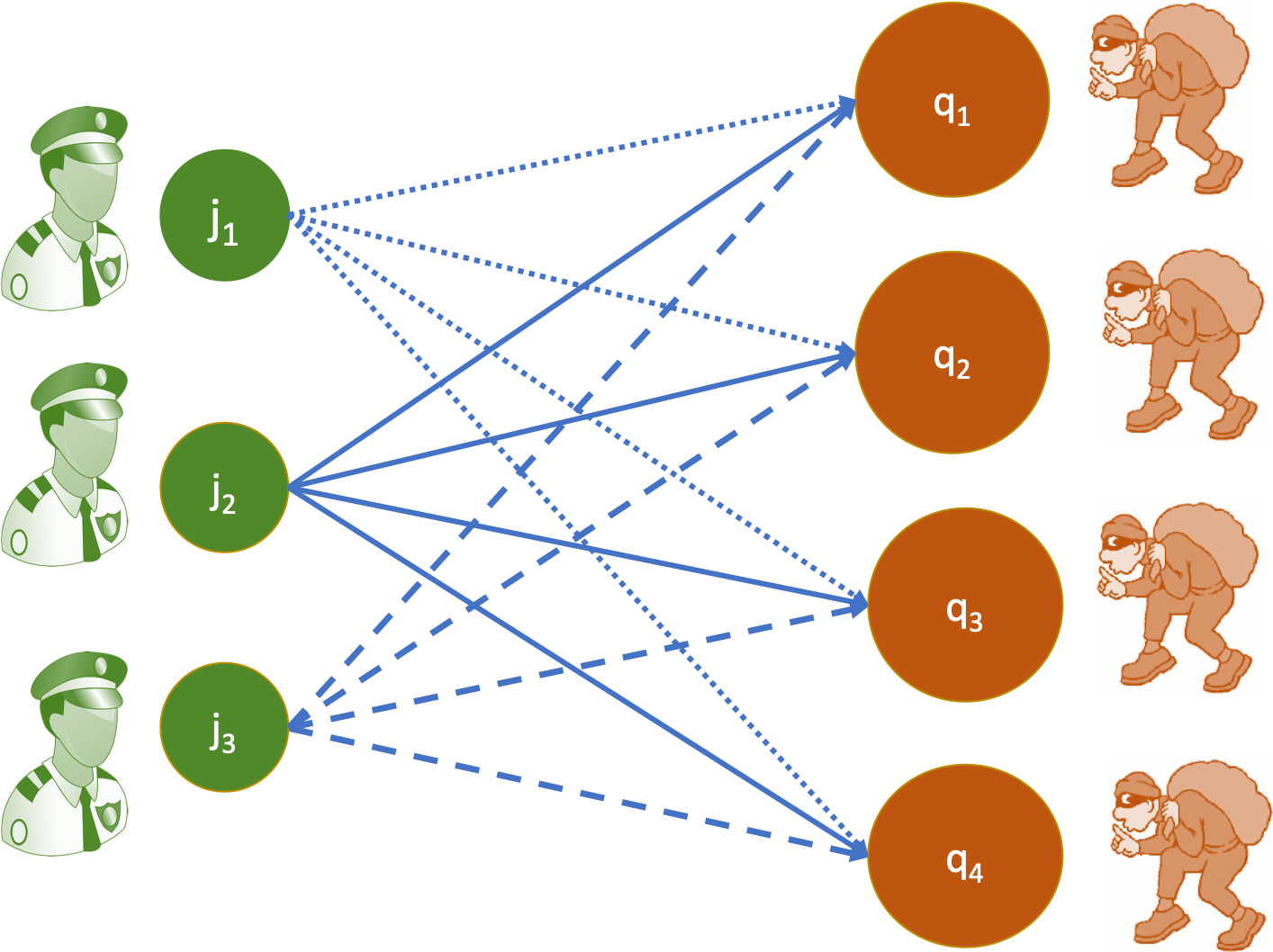}
    \caption{The scheme of resource allocation problem with $m = 3$ police stations (source nodes) and $n = 4$ neighborhoods (target nodes).}
    \label{fig:transport_plan}
\end{figure}

\paragraph{Parameter Setting}
We set the capacities for the source nodes as $c_1 = 4, c_2 =3$ and the coefficient matrix (meaning the cost of transportation along every link) is chosen as follows:
\begin{equation}
 M = \begin{bmatrix}
   1 & 3 & 5 \\
    2 &  5 &  1\end{bmatrix}.
    \label{eq:perception_matrix}
\end{equation}

When there is criminal involved (the Bayesian adversarial OT case), the upper bound matrix $N$ is chosen as 
\begin{equation}
    \underline{n}= \begin{bmatrix}
    6 & 4 & 4   
    \end{bmatrix}, 
    \bar{n} = \begin{bmatrix}
    8 & 10 & 10
    \end{bmatrix}.
\label{eq: criminal_upper_bound}
\end{equation}
The coefficients related to probability of getting caught are 
\begin{equation}
    C = \begin{bmatrix}
    1 & 2 & 3
    \end{bmatrix}.
    \label{punishment_matrix}
\end{equation}
the indices $\beta_1 = 0.5,\beta_2 = 0.5$.
\begin{figure}
    \centering
    \includegraphics[scale=0.5]{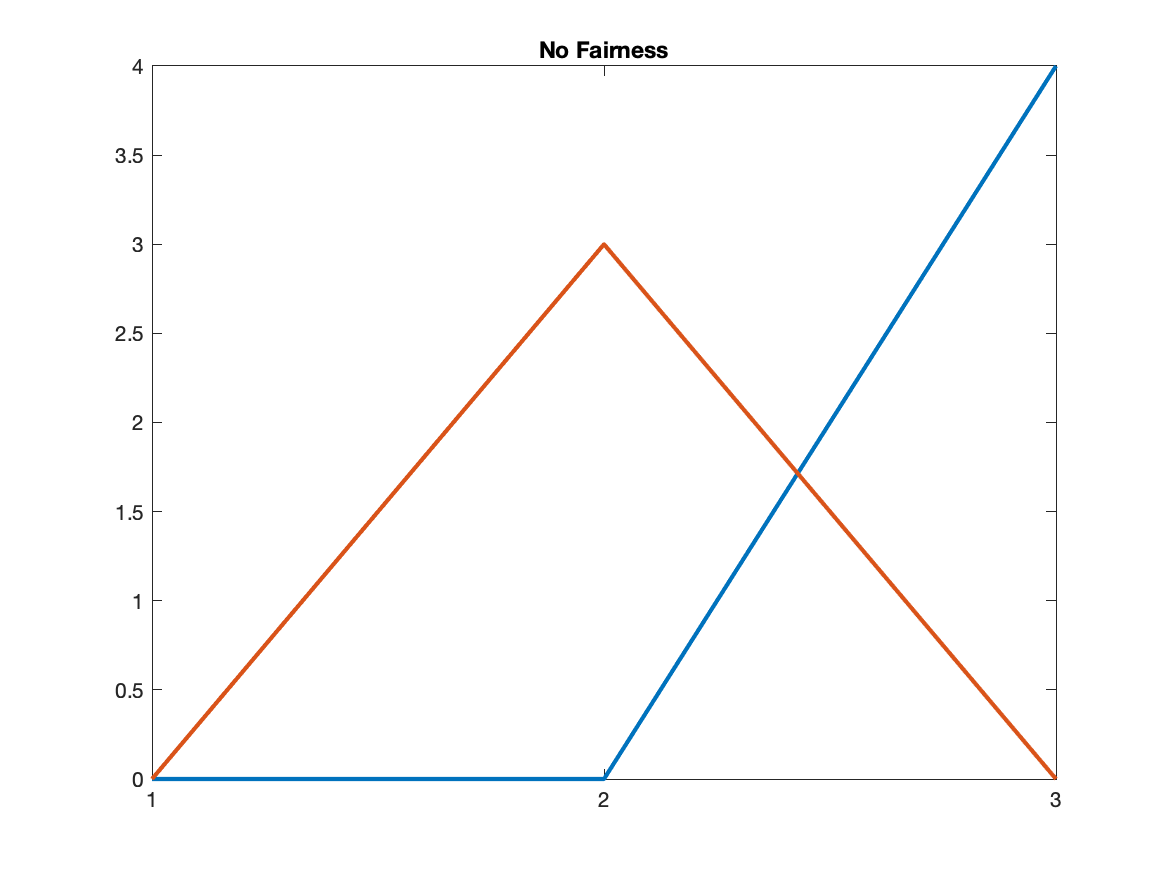}
    \caption{The optimal transport plan of police resource $x^*$ without the fairness term. The dispatcher's utility function is chosen to be \eqref{eq:planner_utility} with $\lambda=0$. We choose the perception matrix $M$ to be in \eqref{eq:perception_matrix}. }
    \label{fig:no_crime_no_fairness}
\end{figure}

\begin{figure}
    \centering
    \includegraphics[scale=0.5]{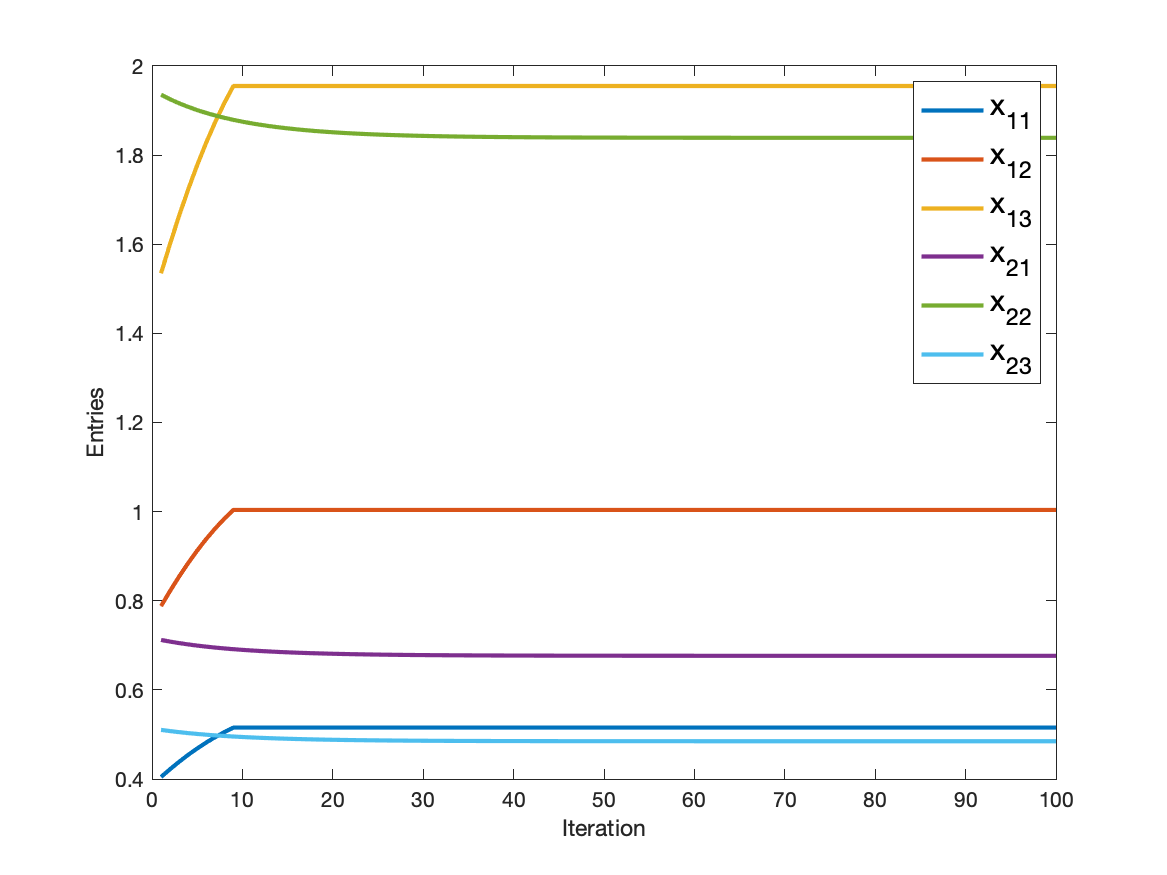}
    \caption{The optimal plan of police resource allocation $x$ with fairness term taken into account. The perception coefficient $M$ is set in \eqref{eq:perception_matrix}. The dispatcher's utility is set in \eqref{eq:planner_utility} with the parameter $\lambda = 3$. }
    \label{fig:no_crime_fairness}
\end{figure}

\begin{figure}
    \centering
    \includegraphics[scale = 0.5]{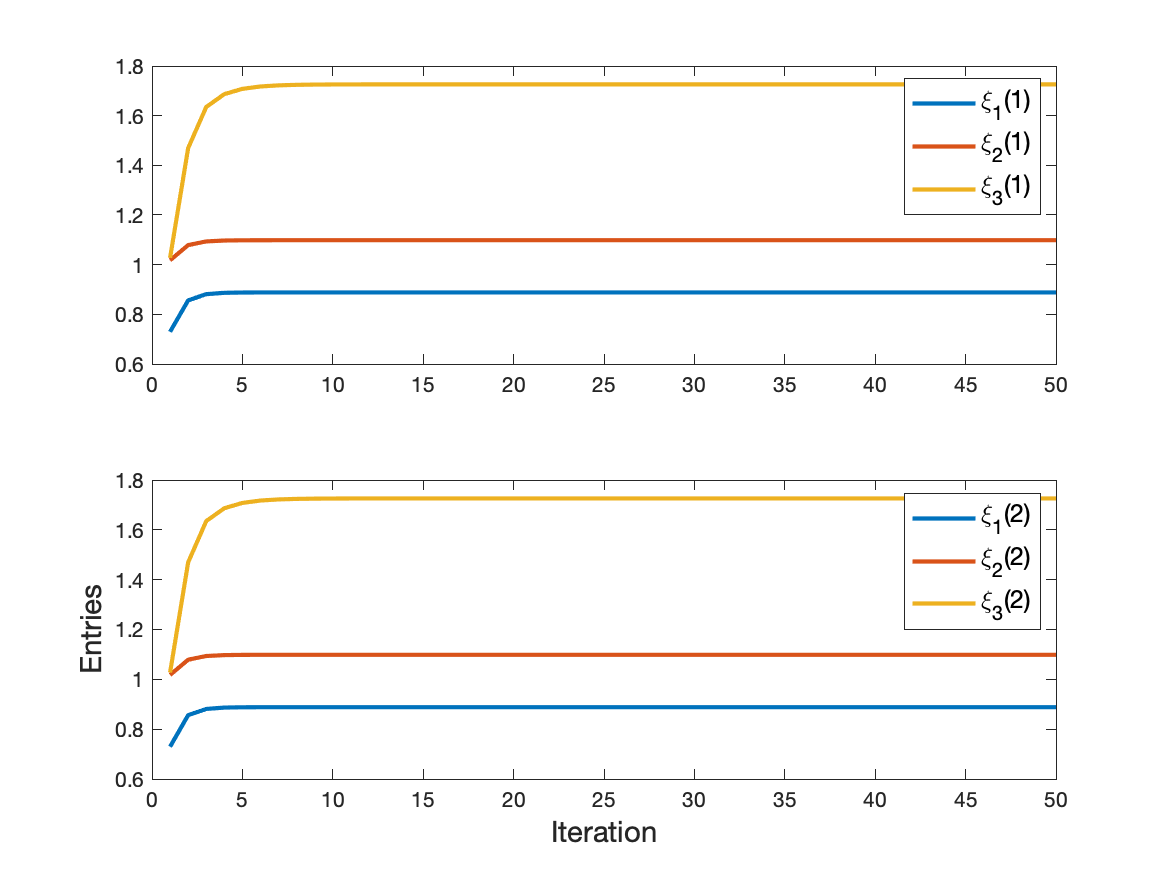}
    \caption{The convergence of criminal actions into equilibrium of the static Bayesian game $\mc{G}$ in definition \ref{def:Bayesian_equilibrium}. We set the perception coefficient matrix $M$ as in \eqref{eq:perception_matrix}, the upper bounds of criminal $\underline{n},\bar{n}$ actions under type $\tha_q = 1, 2$ to be in \eqref{eq: criminal_upper_bound}. The coefficient matrix $C$ corresponding to the punishment is set in \eqref{punishment_matrix}.  }
    \label{fig:crime_fairness_criminal}
\end{figure}
\begin{figure}
    \centering
    \includegraphics[scale=0.5]{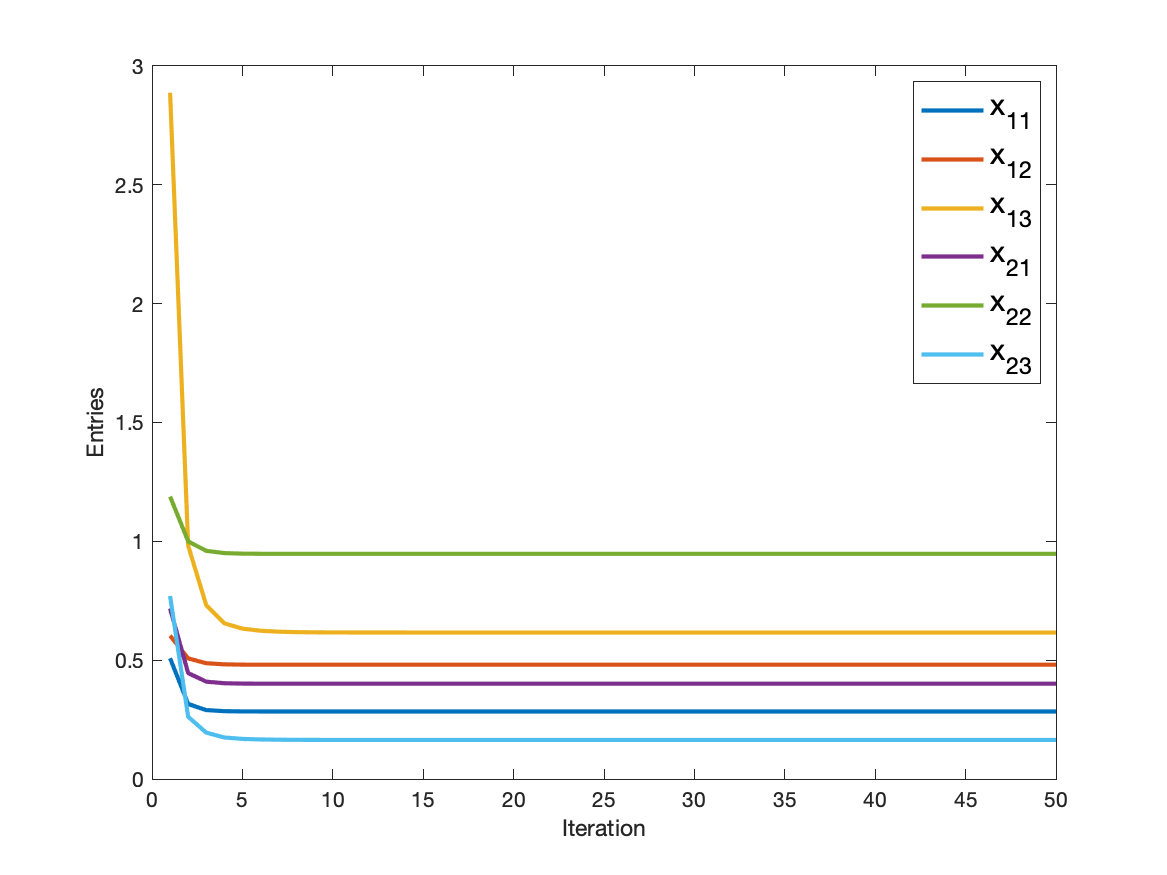}
    \caption{The dispatcher's equilibrium strategies $x^*$ for the static Bayesian adversarial optimal transport game $\mc{G}$. The perception matrix $M$ is given in \eqref{eq:perception_matrix}. The criminal's upper bound activity is given in \eqref{eq: criminal_upper_bound}. The parameter $\lambda=3$. }
 u\label{fig:crime_fairness}
\end{figure}

\begin{figure}
    \centering
    \includegraphics[scale=0.5]{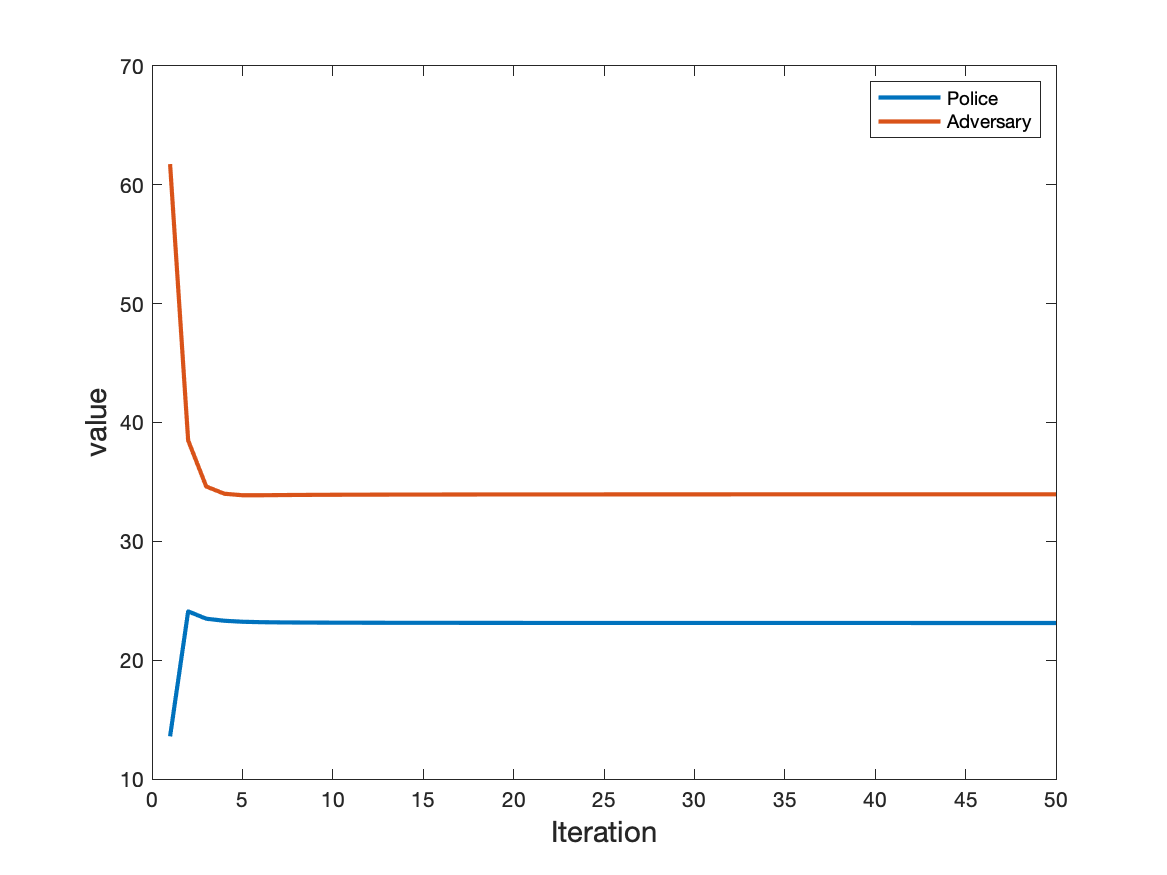}
    \caption{The iterations of the dispatcher's and the criminal's utilities for the Bayesian adversarial optimal transport game $\mc{G}$.}
    \label{fig:util_crime_fairness}
\end{figure}

Since every source node $j$ has two edges connected to it and the data on both edges could be truthful or modified, we could associate the type space of $j$'s adversary as a set containing four elements: $\theta_q = \{\theta_{j_1q},\theta_{j_1q_2}\}$, where $\theta_{j_1q_1},\theta_{j_1q_2}$ could both be $0$ or $1$.  For simplicity we associate a uniform distribution as a common prior over the adversary's type space, that is, the probability that every type is adopted is $\frac{1}{16}$.

1. We first illustrate the impact of the regularization term upon smoothing the transport plans. We set $\lambda=0$ and $\lambda = 3$ in adversary-free scenario in  \eqref{eq:regular_ot} and illustrate the corresponding transport plans in Figure \ref{fig:no_crime_no_fairness} and
Figure \ref{fig:no_crime_fairness}, respectively. We observe that without regularization, every source node tend to distribute all of their resources along the most weighted link they are connected to. With the entropic regularization term, however, every source node begins to distribute their resources more smoothly among its links.  

2. The next example illustrates the influence of the adversary with unknown type upon the social utility of the planner formulated in static games with incomplete information as in Figure \ref{fig:util_crime_fairness}. We compute the Bayesian equilibrium defined in   plot the expected utilities of the dispatcher with and without attack in (a).  We set $\lambda = 3$ and demonstrate the converging transport plans as well as adversary's actions as Bayesian equilibrium of the game in Figure \ref{fig:crime_fairness} and in Figure \ref{fig:crime_fairness_criminal}, respectively. We observe that while the adversary modifies and increase the perception $m$ on some link, he increases the transport rate along that link, which is decided by the dispatcher. 

3. Next we demonstrate how the criminal behavior varies through time.

\section{Conclusion}
\label{sec:conclusion}

We formulate the  problem of seeking an optimal transport plan over networks under the adversarial setting as a Bayesian game where the adversary and the dispatcher have asymmetric information. Specifically, the dispatcher does not know which of the source nodes are malicious and has to develop optimal transport plan based on his subjective probability distribution over the conditions of the source nodes. The asymmetric information happens often in most security setting than complete information. We have also developed a distributed algorithm to solve the Bayesian games, suggesting that large-scale implementation of our formulation is possible since every user can have their own transport rates updated simultaneously without knowing the global topology of the network.

Extensions of the current work can go along several directions. One example is to consider correlated strategies of the dispatchers/the dispatcher and the adversary, while maintaining the assumption that the game is of incomplete information for the dispatchers. Another direction is to consider various information structures for the adversary and the dispatcher. For instance, the dispatcher  Our model also has also the potentiality to be applied large communication networks or urban science context.

\bibliographystyle{plainnat}
\bibliography{rl}

\section*{Appendix A. Placeholder} \label{sec:appendixa}
\addcontentsline{toc}{section}{Appendix A}

\end{document}